\documentclass[onecolumn,12pt,draftclsnofoot]{IEEEtran}

\usepackage{graphicx}
\usepackage{cite}
\usepackage{amsmath}
\usepackage{amsthm}
\usepackage{subfigure}
\usepackage{float}
\usepackage{times}
\usepackage{enumerate}
\usepackage{amssymb}
\usepackage{algorithm}
\usepackage{algorithmic}
\usepackage{multirow}
\usepackage{braket}
\usepackage{multicol}
\usepackage{caption}
\usepackage{color}
\usepackage{setspace}
\newtheorem{theorem}{Theorem}[section]

\newtheorem{proposition}{Proposition}[section]

\IEEEoverridecommandlockouts

\ifCLASSINFOpdf
\else
\fi

\begin{document}

\title{Effect of Energy Harvesting on Stable Throughput in Cooperative Relay Systems}

\author{\IEEEauthorblockN{Nikolaos Pappas, Marios Kountouris, Jeongho Jeon Anthony Ephremides, Apostolos Traganitis\thanks{N. Pappas is with the Department of Science and Technology, Link\"{o}ping University, Norrk\"{o}ping SE-60174, Sweden (e-mail: nikolaos.pappas@liu.se).
M. Kountouris is with the Mathematical and Algorithmic Sciences Lab, France Research Center, Huawei Technologies Co. Ltd. (e-mail: marios.kountouris@supelec.fr).
J. Jeon is with Intel Corporation, Santa Clara, CA 95054 USA (email: jeongho.jeon@gmail.com)
A. Ephremides is with the Department of Electrical and Computer Engineering and Institute for Systems Research, University of Maryland, College Park, MD 20742 (e-mail: etony@umd.edu).
A. Traganitis is with the Computer Science Department, University of Crete, Greece and Institute of Computer Science, Foundation for Research and Technology - Hellas (FORTH) (e-mail: tragani@ics.forth.gr).}
}
\thanks{This work has been partially supported by the People Programme (Marie Curie Actions) of the European Union's Seventh Framework Programme FP7/2007-2013/ under REA grant agreement no.[612361] -- SOrBet and by the MURI grant W911NF-08-1-0238, NSF grant CCF-0728966, ONR grant N000141110127.}
\thanks{This work was presented in part in the 1st IEEE Global Conference on Signal and Information Processing (GlobalSIP) 2013 \cite{b:PappasGlobalsip2013}.}
}

\maketitle

\begin{abstract}
In this paper, the impact of energy constraints on a two-hop network with a source, a relay and a destination under random medium access is studied. A collision channel with erasures is considered, and the source and the relay nodes have energy harvesting capabilities and an unlimited battery to store the harvested energy. Additionally, the source and the relay node have external traffic arrivals and the relay forwards a fraction of the source node's traffic to the destination; the cooperation is performed at the network level. An inner and an outer bound of the stability region for a given transmission probability vector are obtained. Then, the closure of the inner and the outer bound is obtained separately and they turn out to be identical. 
This work is not only a step in connecting information theory and networking, by studying the maximum stable throughput region metric
but also it taps the relatively unexplored and important domain of energy harvesting and assesses the effect of that on this important measure.
\end{abstract}

\IEEEpeerreviewmaketitle

\section{Introduction}

\label{sec:ECOOP_intro}

Taking advantage of renewable energy resources from the environment, also known as energy harvesting, enables unattended operability of infrastructure-less wireless networks. There are various forms of energy that can be harvested, including thermal, solar, acoustic, wind, and even ambient radio power \cite{paradiso:energy}. Energy harvesting is recently seen as a promising feature for wireless networks regarding self-sustainability and also efficiency. This permits long-term operation of distributed wireless communication systems, such as sensor networks, without the need for regular maintenance. However, the additional functionality of energy harvesting in wireless networks introduces several changes and calls for assessment of the system long-term performance such as in terms of the throughput and stability. The ideal scenario is to make the energy limitations transparent to the network.

Among distributed communication protocols, we are particularly interested in ALOHA, a simple random access scheme in which transmission attempts are performed randomly, independently, and distributively \cite{abramson:aloha}. In \cite{jeon:isitstability}, the capability of energy harvesting was first introduced in the analysis of the slotted ALOHA for a simple setting as an initial step to understand its impact on the achievable stability region. Recently, this result has been generalized in \cite{jeon:stability} by taking into account the multi-packet reception capability at the receiver and finite capacity batteries at the energy harvesting sources. In \cite{b:Pappas-JCN}, a cognitive access protocol was studied for the scenario where the higher priority primary source is powered by harvesting energy whereas the lower priority secondary source is assumed to have a reliable power supply.

Cooperative communication is one of key technologies to achieve coverage extension and throughput enhancement in wireless networks \cite{b:Yates-NOW}. In this work, we consider packet-level cooperation rather at the physical layer, in which a relay node takes responsibility of packet delivery for those it could overhear and successfully decode from the transmissions by source node \cite{b:Sadek, b:Rong1, b:Pappas-ISIT}. A key difference between physical-layer and network-layer cooperation is that the latter can capture the bursty nature of traffic. The impact of network-level cooperation in an energy harvesting network with a pure relay (without its own traffic) under scheduled access (time division multiple access in a controlled manner) was studied in \cite{b:Krikidis_Energy_Harv}.

A major limitation of Information Theory is the inability to handle bursty traffic and queuing delay. In the communications networks bursty traffic and
delay are central and indispensable concepts. The information theoretic capacity region is derived under the assumption of saturated queues. However, under stochastic and bursty traffic arrivals, the maximum stable throughput or stability region becomes a meaningful and relevant measure of rates in packets per slot in wireless networks. Thus, the maximum stable throughput is an important performance measure, akin to information theoretic capacity, but simpler to track and analyze and more appropriate for systems with sources that generate signals randomly in time. Understanding the relationship between information-theoretic capacity and stability region has received considerable attention in recent years and some progress has been made primarily for multiple access channels \cite{b:EphremidesHajekUnion}.

The characterization of random access stability for bursty traffic is a challenging problem even without energy harvesting \cite{tsybakov:ergodicity, rao:stability, Szpankowski:stability}. Additionally, for a network with more than three users (interacting queues), the exact characterization of the stability region is not known. This is because each node transmits and, thereby, interferes with the others only when its queue is non-empty. Such queues are said to be \textit{interacting} with each other in the sense that the service process of one depends on the status of the others. The analysis for the case with energy harvesting becomes significantly more challenging because the service process of a node depends not only on the status of its own queue and battery, but also on the status of the other node's queue and battery.

In this paper, we study the impact of energy constraints on a two-hop network with a source, a relay and a destination under random medium access as shown in Fig.~\ref{fig:ECOOP_model}. We assume a collision channel with erasures. Both the source and the relay node have external traffic arrivals. The relay forwards a fraction of the source node's traffic to the destination and the cooperation is performed at the network level. In addition, both source and relay nodes have energy harvesting capabilities and an unlimited battery to store the harvested energy. We provide necessary and sufficient conditions for the stability of the considered network as shown in Fig.~\ref{fig:ECOOP_model}. We first obtain an inner and an outer bound of the stability region for a given transmission probability vector. We then take the closure of the inner and the outer bound separately over all feasible transmission probability vectors. Interestingly, it turns out that the bounds are tight in terms of the closure, as also stated in \cite{jeon:stability}. This study provides insights on designing a relay-assisted network under energy constraints. When the aggregate charging rate is above one and the source and the relay lie in the intermediate traffic regime, the system has identical performance with that of a network without energy constraints, meaning that in that regime the energy limitations are transparent to the network operation. In this paper we focus on a simple network, as mentioned earlier, more realistic and complex systems are impossible to analyze, primarily due to the difficulty in tracking interacting queue. However, insights can still be obtained, even from simple models. This work provides a step in connecting information theory and networking, by studying the maximum stable throughput region metric. Further, it taps the relatively unexplored and important domain of energy harvesting and assesses the effect of that on this important measure.

The rest of this paper is organized as follows. In Section~\ref{sec:ECOOP_model}, we define the stability region, describe the channel model, and explain the packet arrival and energy harvesting models. In Section~\ref{sec:ECOOP_main}, we present inner and outer bounds on the stability region as well as the closure of the stability region. The proofs of our results are given in~\ref{sec:ECOOP_analysis} and \ref{sec:ECOOP_closure}. Finally, we conclude our work in Section~\ref{sec:ECOOP_conclusion}.

\section{System Model}\label{sec:ECOOP_model}
We consider a time-slotted system in which the nodes randomly access a common receiver and both source and relay nodes are powered from randomly time-varying renewable energy sources, as shown in Fig.~\ref{fig:ECOOP_model}. Each node stores the harvested energy in a battery of unlimited capacity. We denote with $S$, $R$, and $D$, the source, the relay and the destination, respectively. Packet traffic originates from both $S$ and $R$, and because of the wireless broadcast nature, $R$ may receive some of the packets transmitted from $S$, which in turn can be relayed to $D$. The packets from $S$ that fail to be received by $D$ but are successfully received by $R$ are relayed by $R$. A half-duplex constraint is imposed here, i.e. $R$ can overhear $S$ only when it is idle.

Each node has an infinite size buffer for storing incoming packets and the transmission of each packet occupies one time slot. Node $R$ has separate queues for the exogenous arrivals and the endogenous arrivals being relayed through $R$. Nevertheless, we can let $R$ have a single queue and merge all arrivals into a single queue as the achievable stable throughput region is not affected~\cite{b:Rong3}. This is due to the fact that the link quality between $R$ and $D$ is independent of which packet is selected for transmission.

The packet arrival and energy harvesting processes at $S$ and $R$ are assumed to be Bernoulli with rates $\lambda_S$, $\delta_S$ and $\lambda_R$, $\delta_R$, respectively, and are independent of each other.
$Q_i$ and $B_i$, $i=S,R$, denote the steady state number of packets and energy units in the queue and the energy source at node $i$, respectively. Furthermore, a node $i$ is called active if both its packet queue and its battery are nonempty at the same time, which is denoted by the event $\mathcal{A}_i = \{\{ B_i \neq 0 \} \cap \{ Q_i \neq 0 \}\}$ and idle otherwise (denoted by $\overline{\mathcal{A}_i}$).
In each time slot, nodes $S$ and $R$ attempt to transmit with probabilities $q_S$ and $q_R$, respectively, whenever they are active.
Decisions on transmission are made independently among the nodes and each transmission consumes one energy unit. We assume a collision channel with erasures in which if both $S$ and $R$ transmit at the same time slot, a collision occurs and both transmissions fail. The probability that a packet transmitted by node $i$ is successfully decoded at node $j (\neq i)$ is denoted by $p_{ij}$, which is the probability that the signal-to-noise ratio (SNR) over the specified link exceeds a certain threshold for successful decoding. These erasure/outage probabilities capture the effect of random fading at the physical layer. The probabilities $p_{SD}$, $p_{RD}$, and $p_{SR}$ denote the success probabilities over the link $S-D$, $R-D$, and $S-R$, respectively. We also assume that node $R$ has a better channel to $D$ than $S$, i.e. $p_{RD} > p_{SD}$.

The cooperation is performed at the protocol (network) level as follows: when $S$ transmits a packet, if $D$ decodes it successfully, it sends an ACK and the packet exits the network; if $D$ fails to decode the packet but $R$ does, then $R$ sends an ACK and takes over the responsibility of delivering the packet to $D$ by placing it in its queue. If neither $D$ nor $R$ decode (or if $R$ does not store the packet), the packet remains in $S$'s queue for retransmission. The ACKs are assumed to be error-free, instantaneous, and broadcasted to all relevant nodes.

The average service rate for the source node is given by
\begin{equation} \label{eqn:service_rate_S}
\begin{aligned}
\mu_S =  \left[ q_S (1-q_R) \mathrm{Pr}\left(B_S \neq 0, \mathcal{A}_R \right) +q_S \mathrm{Pr}(B_S \neq 0,\overline{\mathcal{A}_R} ) \right]
\times \left[p_{SD}+(1-p_{SD})p_{SR} \right],
\end{aligned}
\end{equation}
and for the relay is given by
\begin{equation} \label{eqn:service_rate_R}
\mu_R= \left[q_R (1-q_S) \mathrm{Pr}\left( B_R \neq 0, \mathcal{A}_S \right)+ q_R \mathrm{Pr}(B_R \neq 0,\overline{\mathcal{A}_S} ) \right]\times p_{RD}.
\end{equation}

Denote by $Q_i^t$ the length of queue $i$ at the beginning of time slot $t$. Based on the definition in~\cite{Szpankowski:stability}, the queue is said to be \emph{stable} if
\begin{equation*}\label{eqn:PC_definition_stability}
    \lim_{t \rightarrow \infty} {Pr}[Q_i^t < {x}] = F(x)  \text{ and } \lim_{ {x} \rightarrow \infty} F(x) = 1
\end{equation*}
Loynes' theorem~\cite{b:Loynes} states that if the arrival and service processes of a queue are strictly jointly stationary and the average arrival rate is less than the average service rate, then the queue is stable. If the average arrival rate is greater than the average service rate, then the queue is unstable and the value of $Q_i^t$ approaches infinity almost surely. The stability region of the system is defined as the set of arrival rate vectors $\boldsymbol{\lambda}=(\lambda_1, \lambda_2)$ for which the queues in the system are stable.

\begin{figure}[t]
\centering
\includegraphics[scale=0.8]{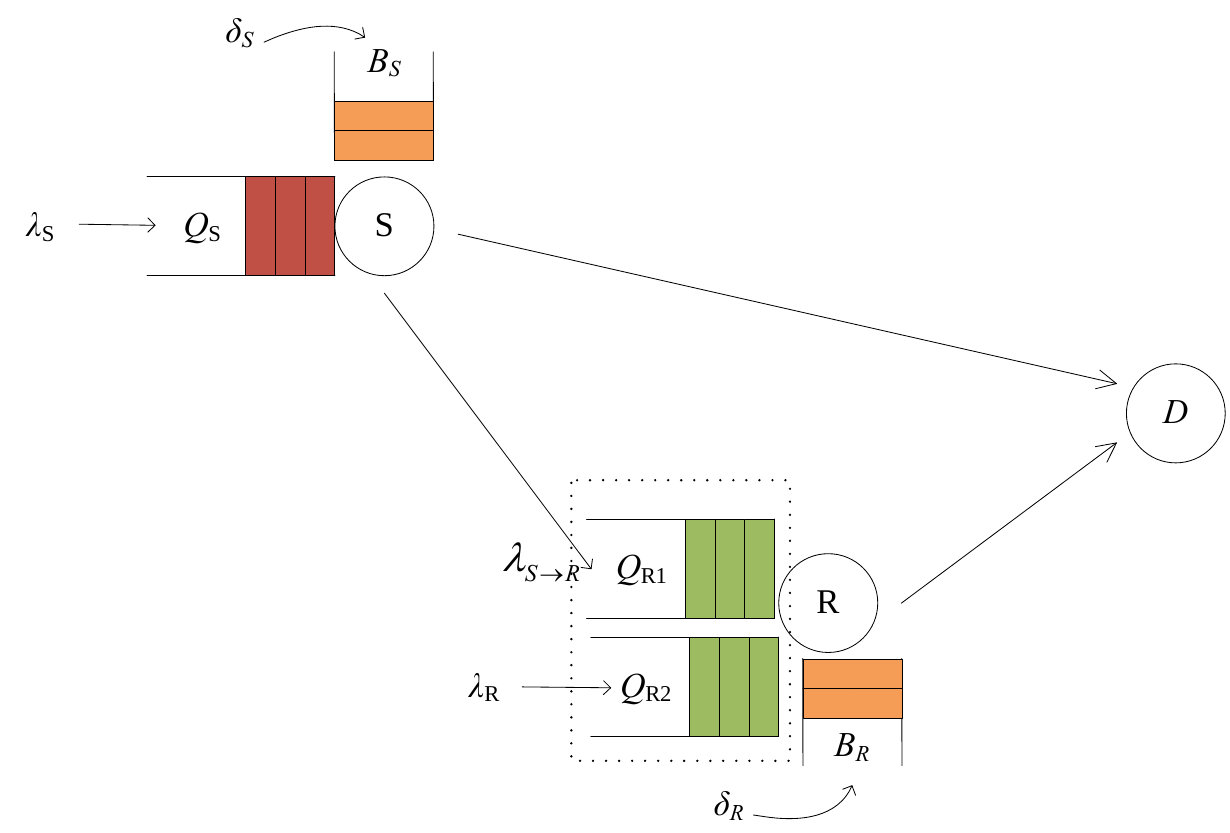}
\caption{A relay-aided wireless network with energy harvesting capabilities.}
\label{fig:ECOOP_model}
\end{figure}

\section{Main Results}\label{sec:ECOOP_main}

This section presents the stability conditions of a network consisting of a source and a relay both having energy harvesting capabilities, and a destination, as depicted in Fig.~\ref{fig:ECOOP_model}.
The source and the relay are assumed to have infinite size queues to store the harvested energy.

The next proposition presents an inner bound on the stability region by providing sufficient conditions for stability.

\begin{proposition} \label{prop:sufficient_conditions}
If $(\lambda_S, \lambda_R) \in \mathcal{R}_{inner}$, with 

\begin{align} \label{eqn:R_inner}
\mathcal{R}_{inner} = \left\{ (\lambda_{S},\lambda_{R}) :  \lambda_S < \min \left(\delta_S ,q_S \right) \left[1 - \min \left(\delta_R ,q_R \right) \right]  \left[p_{SD}+(1-p_{SD})p_{SR} \right] , \right. \notag \\
\left. \lambda_{R}+ \frac{(1-p_{SD})p_{SR}} {p_{SD}+(1-p_{SD})p_{SR}}\lambda_{S}< \min \left(\delta_R ,q_R \right) \left[1 - \min \left(\delta_S ,q_S \right) \right] p_{RD} \right\},
\end{align}
then the network in Fig.~\ref{fig:ECOOP_model} is stable.
\end{proposition}

\begin{proof}
The proof is given in Section~\ref{sec:ECOOP_analysis_suff}.
\end{proof}

The following proposition describes an outer bound of the stability region by obtaining necessary conditions for stability.

\begin{proposition} \label{prop:necessary_conditions}
If the network in Fig.~\ref{fig:ECOOP_model} is stable then $(\lambda_S, \lambda_R) \in \mathcal{R}$, where $ \mathcal{R} = \mathcal{R}_1 \bigcup \mathcal{R}_2$, with 

\begin{align} \label{eqn:R1}
\mathcal{R}_1 =  \left\{ (\lambda_{S},\lambda_{R}) : \left[1+ \frac{\min(\delta_S,q_S) (1-p_{SD}) p_{SR} }{\left[1-\min(\delta_S,q_S)\right] p_{RD}} \right] \lambda_S + \right. \notag \\
\left.  +\frac{\min(\delta_S,q_S) \left[ p_{SD} + (1-p_{SD})p_{SR} \right]}{\left[1-\min(\delta_S,q_S)\right]p_{RD}} \lambda_R < \min(\delta_S,q_S) \left[ p_{SD} + (1-p_{SD})p_{SR} \right] , \right. \notag \\
\left. \lambda_{R}+ \frac{(1-p_{SD})p_{SR}} {p_{SD}+(1-p_{SD})p_{SR}}\lambda_{S} < \min(\delta_R,q_R) \left[1- \min(\delta_S,q_S)\right] p_{RD} \right\}.
\end{align}
\begin{align} \label{eqn:R2}
\mathcal{R}_2 =  \left\{ (\lambda_{S},\lambda_{R}) : \lambda_R + \frac{\left[1-\min(\delta_R,q_R)\right](1-p_{SD})p_{SR}+\min(\delta_R,q_R) p_{RD}}{\left[1-\min(\delta_R,q_R)\right] \left[p_{SD}+(1-p_{SD})p_{SR} \right]} \lambda_S < \min(\delta_R,q_R) p_{RD} , \right. \notag \\
\left. \lambda_S < \min(\delta_S,q_S) \left[1-\min(\delta_R,q_R)\right] \left[p_{SD}+(1-p_{SD})p_{SR} \right] \right\}
\end{align}
\end{proposition}

\begin{proof}
The proof is given in Section~\ref{sec:ECOOP_analysis_necc}.
\end{proof}

\begin{figure}[t]

\centering
 \subfigure[$\mathcal{R}_1$]{
 \includegraphics[scale=0.6]{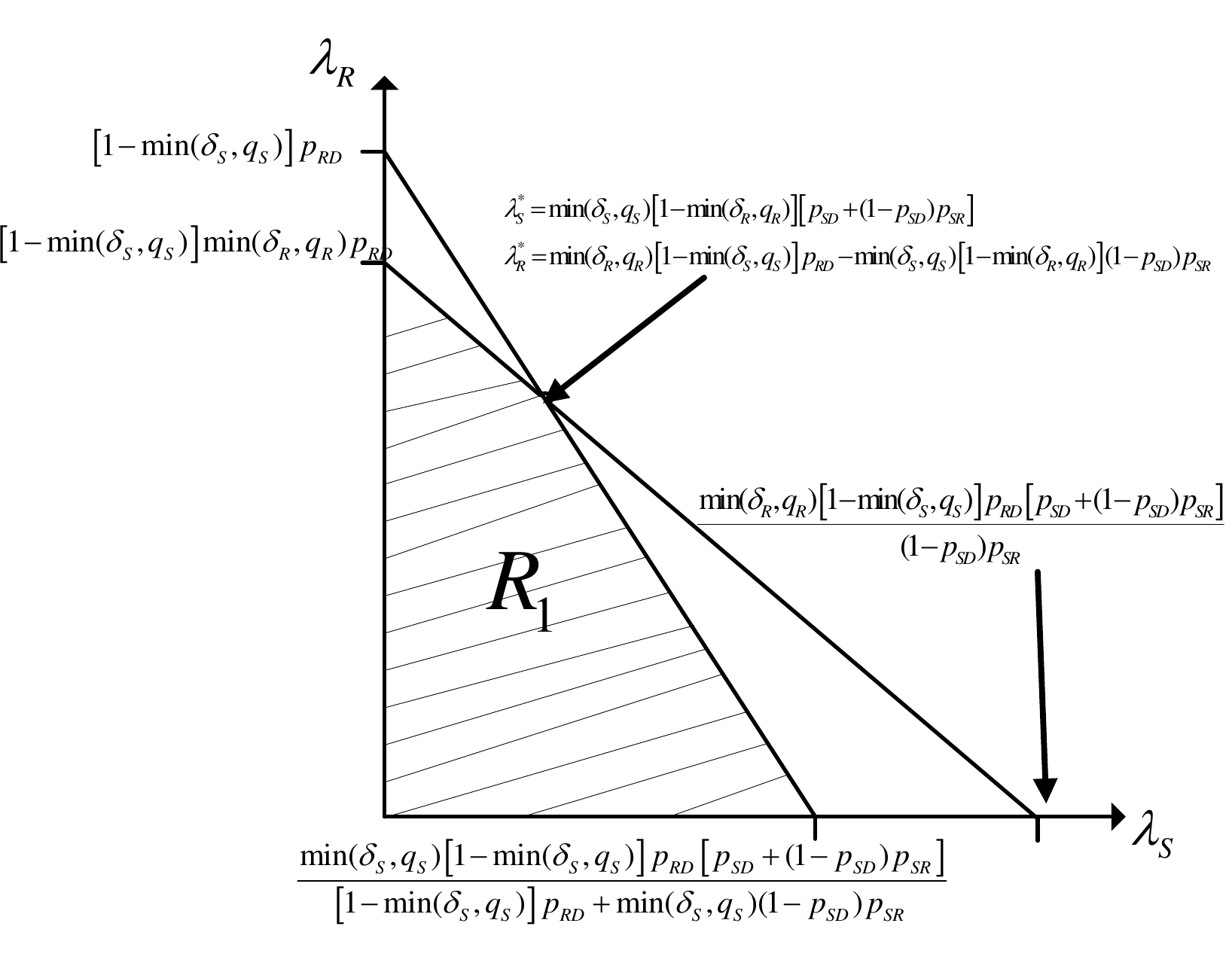}
 \label{fig:ECOOP_1stdom}
 }

  \subfigure[$\mathcal{R}_2$.]{
  \includegraphics[scale=0.6]{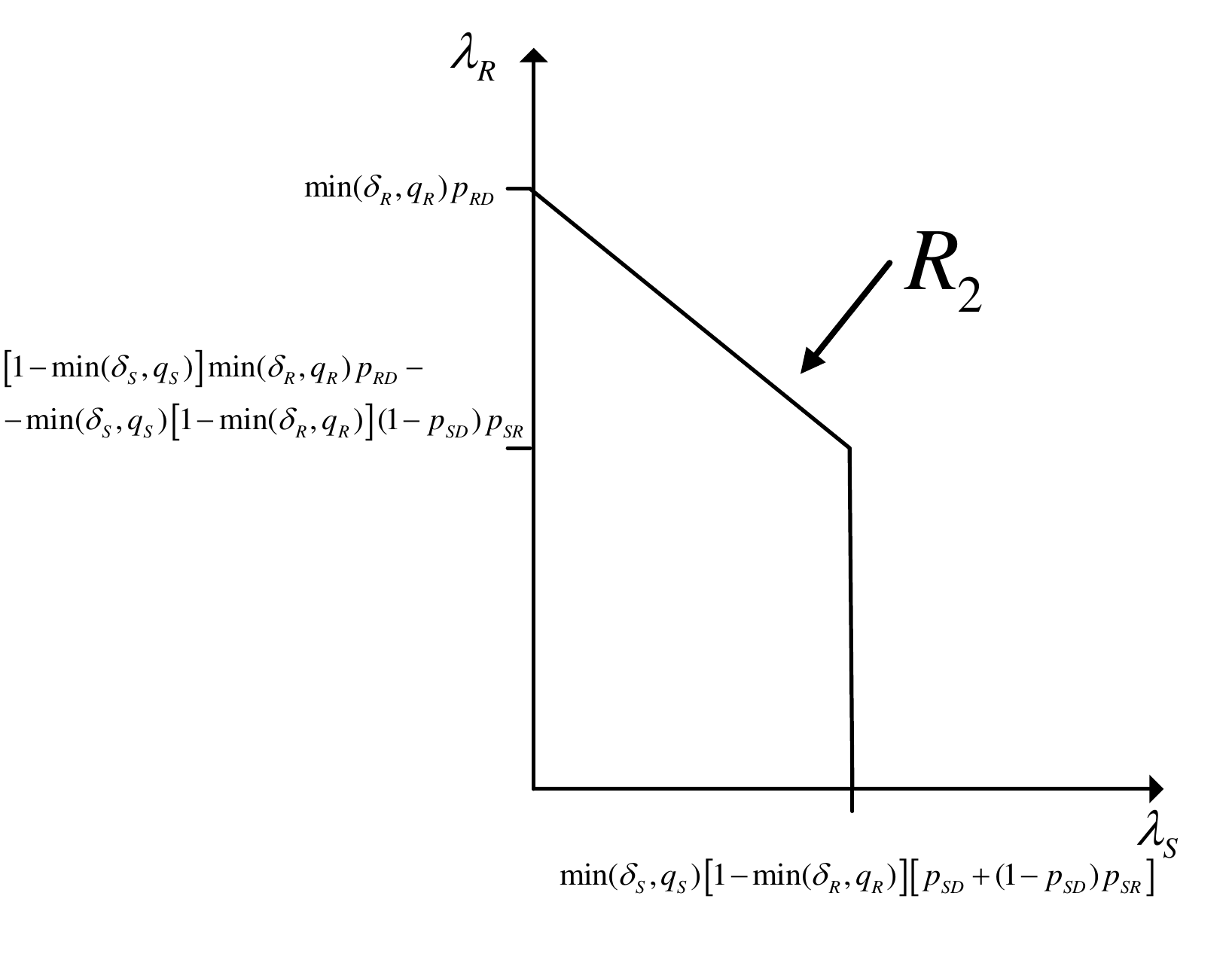}
  \label{fig:ECOOP_2nddom}
  }
  \caption{An outer bound of the stability region $ \mathcal{R} = \mathcal{R}_1 \bigcup \mathcal{R}_2$,  described in Proposition~\ref{prop:necessary_conditions}.}
\end{figure}

Fig.~\ref{fig:ECOOP_1stdom} and~\ref{fig:ECOOP_2nddom} illustrate the $\mathcal{R}_1$ and $\mathcal{R}_2$ described in Proposition~\ref{prop:necessary_conditions}.

In the previous proposition we provided the stability conditions for given transmission probabilities $q_S$ and $q_R$, this stability region is 
denoted by $\mathcal{L}(q_S, q_R, \delta_S, \delta_R)$. Note that

\begin{equation}
\mathcal{R}_{inner} \subseteq \mathcal{L}(q_S, q_R, \delta_S, \delta_R) \subseteq \mathcal{R}_1 \bigcup \mathcal{R}_2.
\end{equation}

The closure of the stability region is defined by

\begin{equation} \label{eq:closure_def}
\mathcal{L}(\delta_S, \delta_R) \triangleq \bigcup_{ (q_S, q_R) \in [0,1]^2} \mathcal{L}(q_S, q_R, \delta_S, \delta_R).
\end{equation}

The following theorem describes the closure of the stability region for the network we consider.

\begin{theorem} \label{thm}
If $\delta_S + \delta_R  \geq 1$, the closure of the stability region, $\mathcal{L}(\delta_S, \delta_R)$, is illustrated in Fig. \ref{fig:closure1} and is described by three parts. (i) The line segment $AB$, where $x_A = 0$, $y_A = \delta_R p_{RD}$ and $x_B=(1 - \delta_R)^2 \left[p_{SD}+(1-p_{SD})p_{SR} \right]$, $y_B = \delta_{R}^2 p_{RD} - (1 - \delta_R)^2 (1-p_{SD})p_{SR}$.

(ii) the curve from $B$ to $C$ which is described by
 
\begin{equation}
\sqrt{\frac{\lambda_S}{P_{SD}+ (1-P_{SD})P_{SR}}} + \sqrt{\frac{P_{SR}(1-P_{SD})\lambda_S}{P_{RD}\left[ P_{SD}+ (1-P_{SD})P_{SR} \right]}+\frac{\lambda_R}{P_{RD}}} = 1
\end{equation}
 
(iii) the line segment $CD$ where $x_C =\delta_S ^2 \left[p_{SD}+(1-p_{SD})p_{SR} \right]$, $y_C=(1-\delta_S)^2 p_{RD}- \delta_S ^2(1-p_{SD})p_{SR}$ and $x_D = \min \left\lbrace \frac{(1-\delta_S)^2 \left[p_{SD}+(1-p_{SD})p_{SR} \right]}{(1-p_{SD})p_{SR}}, \frac{\delta_S(1-\delta_S)p_{RD}\left[p_{SD}+(1-p_{SD})p_{SR} \right]}{(1-\delta_S)p_{RD}+\delta_S (1-p_{SD})p_{SR}} \right\rbrace $, $y_D=0$.

If $\delta_S + \delta_R  < 1$, the closure of the stability region, $\mathcal{L}(\delta_S, \delta_R)$, is illustrated in Fig. \ref{fig:closure2} and is described by the line segments $EF$ and $FG$, where $x_E=0$, $y_E=\delta_R p_{RD}$, $x_F=\delta_S (1-\delta_R) \left[p_{SD}+(1-p_{SD})p_{SR} \right]$, $y_F= \delta_R (1-\delta_S) p_{RD} - \delta_S (1-\delta_R) (1-p_{SD}) p_{SR}$, 

$x_G=\min \left\lbrace \frac{(1-\delta_S) \delta_R p_{RD} \left[p_{SD}+(1-p_{SD})p_{SR} \right]}{(1-p_{SD})p_{SR}}, \frac{\delta_S(1-\delta_S)p_{RD}\left[p_{SD}+(1-p_{SD})p_{SR} \right]}{(1-\delta_S)p_{RD}+\delta_S (1-p_{SD})p_{SR}} \right\rbrace$ and $y_G=0$.
 \end{theorem}

\begin{proof}
The proof is given in Section~\ref{sec:ECOOP_closure}.
\end{proof}

\begin{figure}[t]
\centering
\includegraphics[scale=0.5]{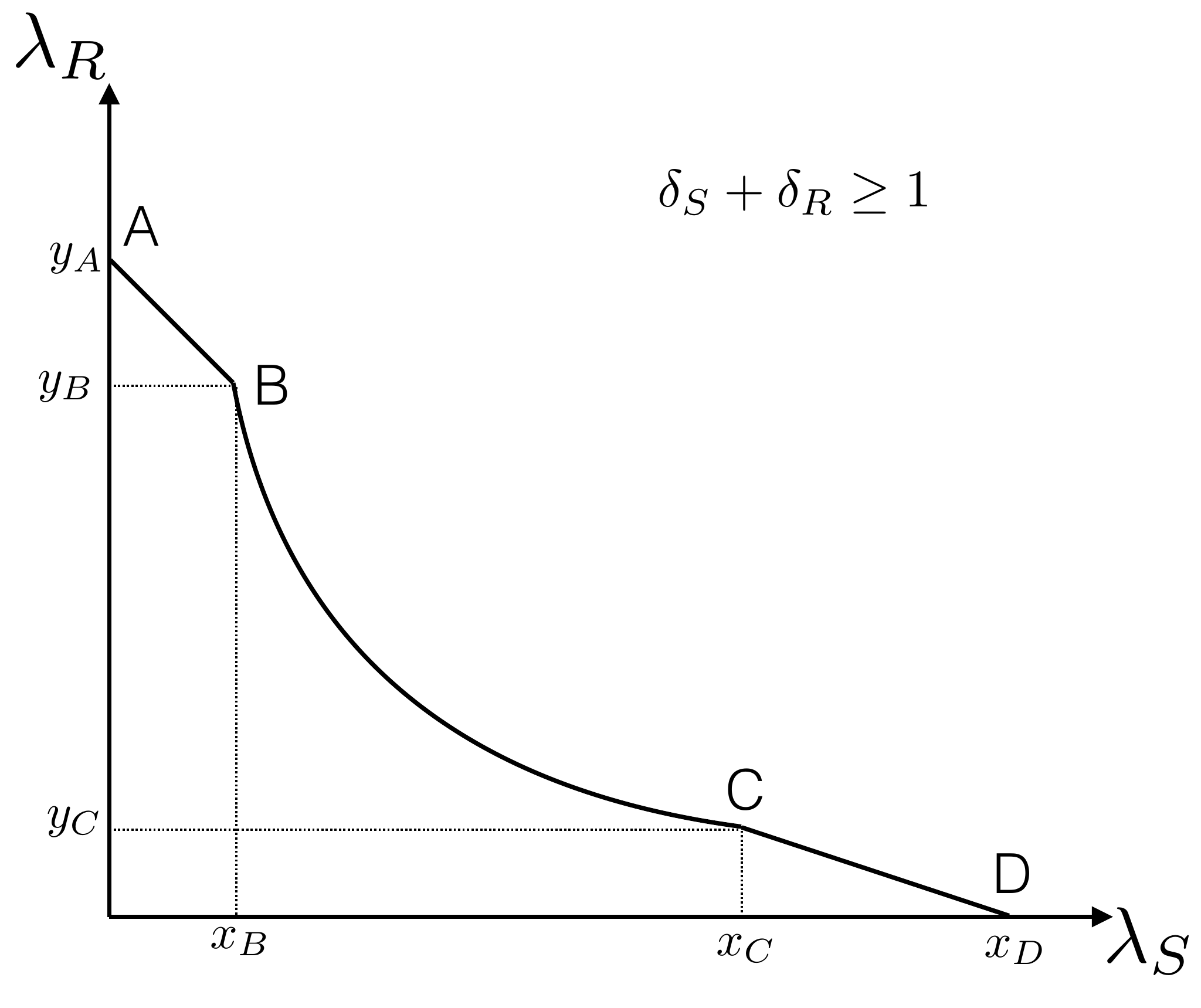}
\caption{The closure of the stability region for $\delta_S+\delta_R \geq 1$.}
\label{fig:closure1}
\end{figure}

\begin{figure}[t]
\centering
\includegraphics[scale=0.5]{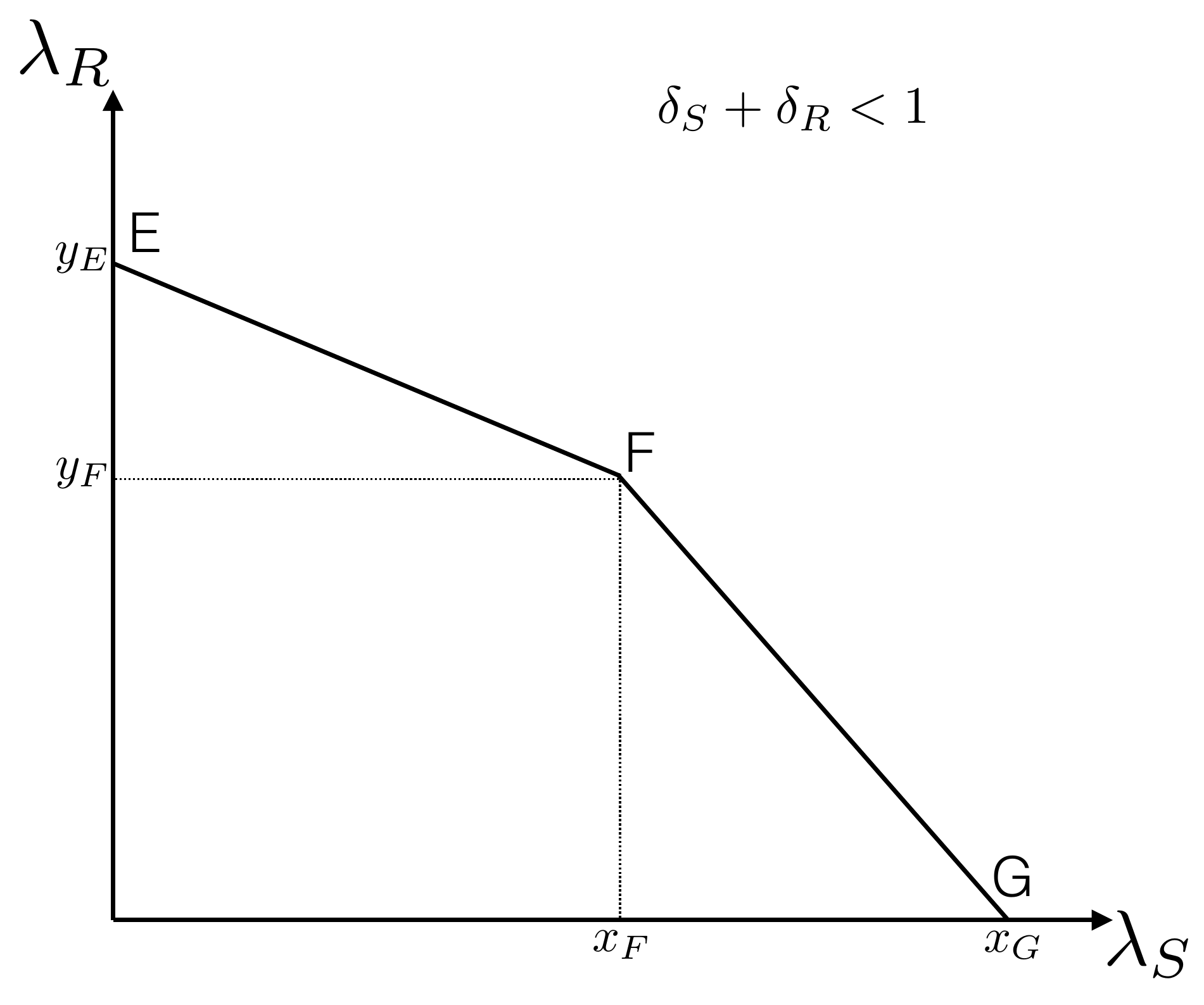}
\caption{The closure of the stability region for $\delta_S+\delta_R < 1$.}
\label{fig:closure2}
\end{figure}

\emph{Remark 1:} Regarding the stability region $\mathcal{L}(q_S, q_R, \delta_S, \delta_R)$, we only have inner and outer bounds and not the exact expression, however
for the closure $\mathcal{L}(\delta_S, \delta_R)$, we do have the exact characterization.

\emph{Remark 2:} If $\delta_S + \delta_R  < 1$ as depicted in Fig. \ref{fig:closure2}, the closure of the stability region has a linear behavior, which is an indication that the performance of the system is affected by the low energy harvesting rates. Furthermore when $\delta_S + \delta_R > 1$, when the arrival
rate at the source or the relay is at high arrival rate regime then the performance is affected by the energy harvesting rate and is depicted by the linear segments in Fig. \ref{fig:closure1}. The interesting case is when both the arrival rates $\lambda_S$ and $\lambda_R$ lie in the intermediate arrival rate regime, then the performance is identical to the relay network without energy limitations and the closure of the stability region has a non-linear behavior.

\section{Analysis} \label{sec:ECOOP_analysis}

To derive the stability condition for the queue in the relay node, we need to calculate the total arrival rate. There are two independent arrival processes at the relay: the exogenous traffic with arrival rate $\lambda_R$ and the endogenous traffic from $S$. Denote by $S_A$ the event that $S$ transmits a packet and the packet leaves the queue, then
\begin{equation}
\mathrm{Pr}(S_A)=\left[1-q_{R}\mathrm{Pr}(\mathcal{A}_R)  \right] \left[p_{SD}+(1-p_{SD})p_{SR} \right].
\end{equation}

Among the packets that depart from the queue of $S$, some will exit the network because they are decoded by the destination directly, and some will be relayed by $R$. Denote by $S_B$ the event that the transmitted packet from $S$ will be relayed from $R$, then
\begin{equation}
\mathrm{Pr}(S_B)=\left[1-q_{R}\mathrm{Pr}(\mathcal{A}_R) \right] (1-p_{SD})p_{SR}.
\end{equation}
The conditional probability that a transmitted packet from $S$ is relayed by $R$ given that the transmitted packet exits node $S$'s queue is given by
\begin{equation}
\mathrm{Pr}(S_B | S_A)=\frac{(1-p_{SD})p_{SR}} {p_{SD}+(1-p_{SD})p_{SR}}.
\end{equation}
The arrival rate from the source to the relay is
\begin{equation}
\lambda_{S \rightarrow R} = \mathrm{Pr}(S_B | S_A) \lambda_S.
\end{equation}

Note that the average arrival rate from the source to the relay, $\lambda_{S \rightarrow R}$, does not depend on the transmission probabilities or the energy harvesting rates, it depends only from the success probabilities $p_{SD}$ and $p_{SR}$.

The total arrival rate at the relay node is given by
\begin{equation}
\lambda_{R,total}=\lambda_{R}+ \frac{(1-p_{SD})p_{SR}} {p_{SD}+(1-p_{SD})p_{SR}}\lambda_{S}.
\end{equation}

\subsection{Sufficient Conditions} \label{sec:ECOOP_analysis_suff}

A queue is considered saturated if in each time slot there is always a packet to transmit, i.e. the queue is never empty.
Assuming saturated queues for the source and the relay node, the saturated throughput for the source node is given by
\begin{equation}
\begin{aligned}
\mu_{S}^{s} = \left\{ q_S (1-q_R) \mathrm{Pr}\left(B_S\neq 0, B_R \neq 0 \right) + q_S  \mathrm{Pr} \left(B_S \neq 0, B_R = 0 \right) \right\}  \left[p_{SD}+(1-p_{SD})p_{SR} \right],
\end{aligned}
\end{equation}

and for the relay is given by
\begin{align}
\mu_{R}^{s} = \left\{q_R (1-q_S) \mathrm{Pr}\left(B_S \neq 0, B_R \neq 0 \right)+  q_R \mathrm{Pr}(B_S = 0 ,B_R \neq 0) \right\} p_{RD}.
\end{align}

Each node transmits with probability $q_i$, $i=S,R$, whenever its battery is not empty and each transmission demands one energy packet. Each energy queue $i$ is then decoupled and forms a discrete-time $M/M/1$ queue with input rate $\delta_i$
and service rate $q_i$, thus the probability the energy queue to be empty is given by
\begin{equation} \label{eqn:prob_nonempty_battery}
\mathrm{Pr}\left(B_i \neq 0 \right) = \min \left( \frac{\delta_i}{q_i} , 1 \right).
\end{equation}

Note that the battery queues are decoupled since we assume that each node, either the source or the relay, attempts to transmit independently in a random access manner.

Then, after some calculations, we obtain that the saturated throughput for the source is
\begin{equation}
\mu_{S}^{s} = \min \left(\delta_S ,q_S \right) \left[1 - \min \left(\delta_R ,q_R \right) \right]  \left[p_{SD}+(1-p_{SD})p_{SR} \right],
\end{equation}
and for the relay is
\begin{equation} \label{eqn:saturated_thr_R}
\mu_{R}^{s} = \min \left(\delta_R ,q_R \right) \left[1 - \min \left(\delta_S ,q_S \right) \right] p_{RD}.
\end{equation}

The sufficient conditions ($\mathcal{R}_{inner}$) for the stability are obtained by $\lambda_S < \mu_{S}^{s}$ and $\lambda_{R,total} <\mu_{R}^{s}$ and are given by (\ref{eqn:R_inner}), in Proposition~\ref{prop:sufficient_conditions}.

The saturated throughput that we obtained in this subsection is an inner bound of the stability region, since the assumption that the source and the relay have always packets to transmit leads to lower achievable rates in terms of packet per slot \cite{b:EphremidesHajekUnion, b:VerduAnantharam, b:Gallager76}.

\subsection{Necessary Conditions} \label{sec:ECOOP_analysis_necc}

The average service rates for the source and the relay are given by (\ref{eqn:service_rate_S}) and (\ref{eqn:service_rate_R}), respectively. The average service rate of each queue depends on the status of its own energy and also the
queue size and the energy statues of the other queues. This coupling between the queues (both packet and energy) results in a four dimensional Markov chain which makes the analysis cumbersome. Therefore, the stochastic dominant technique~\cite{rao:stability} is essential in order to decouple the interaction between the queues, and thus to characterize the stability region. Thus, we first construct parallel dominant systems in which one of the nodes transmits dummy packets when its packet queue is empty. Note that even in the dominant system a node cannot transmit if the energy source is empty (because even the dummy packet consumes one energy unit).

We consider the first hypothetical system in which the source node transmits dummy packets when its queue is empty and all the other assumptions remain intact.
The average service rate for the relay given by (\ref{eqn:service_rate_R}) becomes
\begin{equation}
\mu_R= \left\{q_R (1-q_S) \mathrm{Pr}\left(B_S \neq 0, B_R \neq 0 \right)+ q_R \mathrm{Pr}(B_S = 0 ,B_R \neq 0) \right\} p_{RD}.
\end{equation}

The average service rate of the relay, $\mu_{R}$, in the first hypothetical system is the same with the saturated throughput of the relay obtained in (\ref{eqn:saturated_thr_R}). From Loyne's criterion, the relay is stable if $\lambda_{R,total} < \mu_R$, thus
\begin{equation}
\lambda_{R}+ \frac{(1-p_{SD})p_{SR}} {p_{SD}+(1-p_{SD})p_{SR}}\lambda_{S} < \min \left(\delta_R ,q_R \right) \left[1 - \min \left(\delta_S ,q_S \right) \right] p_{RD}.
\end{equation}

The average number of packets per active slot for $R$ is $\left[1 - \min \left(\delta_S ,q_S \right) \right] q_R p_{RD}$, thus the fraction of active slots is given by
\begin{equation} \label{eqn:prob_active_R_1st_hyp_system}
 \mathrm{Pr}\left(B_R \neq 0, Q_R\neq 0 \right) = \frac{\lambda_{R}+ \frac{(1-p_{SD})p_{SR}} {p_{SD}+(1-p_{SD})p_{SR}}\lambda_{S}}{\left[1-\min \left(\delta_S ,q_S \right) \right] q_R p_{RD}}.
\end{equation}

After changing (\ref{eqn:prob_nonempty_battery}) and (\ref{eqn:prob_active_R_1st_hyp_system}) into (\ref{eqn:service_rate_S}), the service rate for the source becomes
\begin{equation}
\begin{aligned}
\mu_S= \min \left(\delta_S ,q_S \right) \left[1 - \frac{\lambda_{R}+ \frac{(1-p_{SD})p_{SR}} {p_{SD}+(1-p_{SD})p_{SR}}\lambda_{S}}{\left[1-\min \left(\delta_S ,q_S \right) \right] p_{RD}} \right] \left[p_{SD}+(1-p_{SD})p_{SR} \right].
\end{aligned}
\end{equation}

The queue in $S$ is stable if $\lambda_S < \mu_S$ and after some manipulations we obtain
\begin{equation}
\begin{aligned}
\left[1+ \frac{\min \left(\delta_S ,q_S \right) (1-p_{SD}) p_{SR} }{\left[1-\min \left(\delta_S ,q_S \right) \right]p_{RD}} \right] \lambda_S + \frac{\min \left(\delta_S ,q_S \right) \left[ p_{SD} + (1-p_{SD})p_{SR} \right]}{
\left[1-\min \left(\delta_S ,q_S \right) \right]p_{RD}} \lambda_R \\
 < \min \left(\delta_S ,q_S \right) \left[ p_{SD} + (1-p_{SD})p_{SR} \right].
\end{aligned}
\end{equation}

The derived stability conditions from the first hypothetical system are summarized in (\ref{eqn:R1}).

In the second hypothetical system, the relay node transmits dummy packets and all the other assumptions remain intact. Thus, the average service rate for the source given by (\ref{eqn:service_rate_S}) becomes
\begin{equation}
\begin{aligned}
\mu_S= \left\{ q_S (1-q_R) \mathrm{Pr}\left(B_S \neq 0, B_R \neq 0 \right) +q_S \mathrm{Pr}(B_S \neq 0, B_R = 0) \right\}  \left[p_{SD}+(1-p_{SD})p_{SR} \right],
\end{aligned}
\end{equation}

which is equal to saturated throughput of the source and is given by
\begin{equation}
\begin{aligned}
\mu_S=\min \left(\delta_S ,q_S \right) \left[1 - \min \left(\delta_R ,q_R \right) \right]  \left[p_{SD}+(1-p_{SD})p_{SR} \right].
\end{aligned}
\end{equation}

From Loyne's theorem, the queue in source is stable if $\lambda_S < \mu_S$, thus
\begin{equation}
\begin{aligned}
\lambda_S < \min \left(\delta_S ,q_S \right) \left[1 - \min \left(\delta_R ,q_R \right) \right] \left[p_{SD}+(1-p_{SD})p_{SR} \right].
\end{aligned}
\end{equation}

The average number of packets per active slot for $S$ is $ q_S \left[1 - \min \left(\delta_R ,q_R \right) \right] \left[p_{SD}+(1-p_{SD})p_{SR} \right] $.
The fraction of active slots for the source $S$ is
\begin{equation} \label{eqn:prob_active_S_2nd_hyp_system}
 \mathrm{Pr}\left(B_S \neq 0, Q_S\neq 0 \right) = \frac{\lambda_{S}}{q_S \left[1 - \min \left(\delta_R ,q_R \right) \right] \left[p_{SD}+(1-p_{SD})p_{SR} \right]}.
\end{equation}

After replacing from (\ref{eqn:prob_nonempty_battery}) and (\ref{eqn:prob_active_S_2nd_hyp_system}) into (\ref{eqn:service_rate_R}), the service rate for the relay is
\begin{equation}
\mu_R=\min \left(\delta_R ,q_R \right) \left[1- \frac{\lambda_{S}}{\left[1 - \min \left(\delta_R ,q_R \right) \right] \left[p_{SD}+(1-p_{SD})p_{SR} \right]}\right] p_{RD}.
\end{equation}

The queue in the relay node $R$ is stable if $\lambda_{R,total} < \mu_R$ and after some manipulations we obtain
\begin{equation}
\begin{aligned}
\lambda_R + \frac{\left[1 - \min \left(\delta_R ,q_R \right) \right](1-p_{SD})p_{SR}+\min \left(\delta_R ,q_R \right) p_{RD}}{\left[1 - \min \left(\delta_R ,q_R \right) \right] \left[p_{SD}+(1-p_{SD})p_{SR} \right]} \lambda_S 
< \min \left(\delta_R ,q_R \right) p_{RD}.
\end{aligned}
\end{equation}

The derived stability conditions from the second hypothetical system are given by (\ref{eqn:R2}).

An important observation made in \cite{rao:stability} is that the stability conditions obtained by using the stochastic dominance technique are not merely sufficient conditions for the stability of the original system but are sufficient and necessary conditions. However, the \emph{indistinguishability} argument does not apply to our problem. In a system with batteries, the dummy packet transmissions affect the dynamics of the batteries. For example, there are instants when a node is no more able to transmit in the hypothetical system
because of the lack of energy, while it is able to transmit in the original system, thus it may result to a better chance of success for the other node.

The obtained stability conditions are necessary conditions of the original system and are summarized in Proposition~\ref{prop:necessary_conditions}.

\section{Proof of Theorem \ref{thm}} \label{sec:ECOOP_closure}
In this section we will derive the closure of the outer of the stability region defined in Proposition \ref{prop:necessary_conditions}.
The closure will be obtained over all the feasible transmission probability vectors $(q_S,q_R) \in [0,1]^2$ and is defined by (\ref{eq:closure_def}). After that we will show that the closure of the outer bound can be achieved. 

An interesting observation is that the outer bound in Proposition \ref{prop:necessary_conditions} does not depend on $\delta_i$, $i=S,R$ for $q_S \leq \delta_S$ and $q_R \leq \delta_R$. Furthermore, if $q_i$ is increased over $\delta_i$ for $i=S,R$ has no effect because the value of $\min (\delta_i, q_i)$ is bounded by $\delta_i$.

In order to obtain the closure we have to solve two optimization problems. By replacing $\lambda_S$ by $x$ and $\lambda_R$ by $y$, the optimization problems are
\begin{equation}\label{eq:P1}
\text{[$P_1$]       }   \max_{q_S} \text{   }x=\frac{q_S (1-q_S) \left[P_{SD}+ (1-P_{SD})P_{SR}\right] P_{RD}}{(1-q_S)P_{RD}+q_S (1-P_{SD})P_{SR}} - \frac{q_S \left[P_{SD}+ (1-P_{SD})P_{SR}\right]}{(1-q_S)P_{RD}+q_S (1-P_{SD})P_{SR}}y
\end{equation}
\begin{eqnarray}
\text{subject to    } & y+\frac{(1-P_{SD})P_{SR}}{\left[P_{SD}+ (1-P_{SD})P_{SR}\right]}x<q_R (1-q_S) P_{RD} \\
& (q_S, q_R) \in [0,\delta_S] \times [0,\delta_R] \\
& (x, y) \in [0,1]^2
\end{eqnarray}
and
\begin{equation} \label{eq:P2}
\text{[$P_2$]       }   \max_{q_R} \text{   }y=q_R P_{RD} -\frac{(1-P_{SD})P_{SR}}{\left[P_{SD}+ (1-P_{SD})P_{SR}\right]}x - \frac{q_R P_{RD}}{(1-q_R) \left[P_{SD}+ (1-P_{SD})P_{SR}\right]} x
\end{equation}
\begin{eqnarray}
\text{subject to    } & x < q_S (1-q_R) \left[P_{SD}+ (1-P_{SD})P_{SR}\right] \\
&(q_S, q_R) \in [0,\delta_S] \times [0,\delta_R] \\
&(x, y) \in [0,1]^2
\end{eqnarray}

In order to solve $[P_2]$, the differentiation of $y$ with respect to $q_R$ gives
\begin{equation}
\frac{dy}{dq_R}=P_{RD} \left[ 1- \frac{x}{(1-q_R)^2  \left[P_{SD}+ (1-P_{SD})P_{SR}\right] } \right].
\end{equation}

The second derivative is negative since
\begin{equation}
\frac{d^2y}{dq^2_R}=- \frac{2 P_{RD} x}{(1-q_R)^3 \left[P_{SD}+ (1-P_{SD})P_{SR}\right]} < 0,
\end{equation}

thus, the objective function is concave with respect to $y$. Solving the equation $\frac{dy}{dq_R}=0$, we obtain the maximizing $q_{R}^{*}$ where
\begin{equation}
q_{R}^{*}=1-\sqrt{\frac{x}{P_{SD}+ (1-P_{SD})P_{SR}}},
\end{equation}
the corresponding maximum value of the objective function is
\begin{equation}
y^* = P_{RD} - 2 P_{RD} \sqrt{\frac{x}{P_{SD}+ (1-P_{SD})P_{SR}}} + \frac{x}{P_{SD}+ (1-P_{SD})P_{SR}} \left[P_{RD} - (1-P_{SD})P_{SR} \right].
\end{equation}

Suppose that $q_{R}^* \in (0,\delta_R)$, then
\begin{equation}\label{eq:p2cr1}
(1-\delta_R)^2 \left[ P_{SD}+ (1-P_{SD})P_{SR} \right] < x <  \left[ P_{SD}+ (1-P_{SD})P_{SR} \right],
\end{equation}

but since $x < q_S (1-q_R)  \left[ P_{SD}+ (1-P_{SD})P_{SR} \right]$ and $q_S \in [0,\delta_S]$ then
\begin{equation}\label{eq:p2cr2}
x \leq \delta_S ^ 2  \left[ P_{SD}+ (1-P_{SD})P_{SR} \right].
\end{equation}

We need to find the intersection of (\ref{eq:p2cr1}) and (\ref{eq:p2cr2}). If $\delta_S + \delta_R < 1$ the intersection
is an empty set, if $\delta_S + \delta_R \geq 1$ then
\begin{equation}\label{eq:p2cr3}
(1-\delta_R)^2 \left[ P_{SD}+ (1-P_{SD})P_{SR} \right] < x \leq \delta_S ^2 \left[ P_{SD}+ (1-P_{SD})P_{SR} \right].
\end{equation}

Suppose that the $q_R^*=0$ or $q_R^*=\delta_R$, which is the case that $x$ lies outside (\ref{eq:p2cr3}).
If $x$ lies outside (\ref{eq:p2cr3}) on the right side then $\frac{dy}{dq_R}$ is always non-positive and $y$ is a non-increasing function of $q_R$,
thus $q_R^*=0$ and $y^* < 0$.

On the other hand if $x$ is on the left side $x \leq (1-\delta_R)^2 \left[ P_{SD}+ (1-P_{SD})P_{SR} \right]$, then $\frac{dy}{dq_R}$ is always non-negative
and thus $y$ is a non-decreasing function $q_R$, hence $q_R^*=\delta_R$ and the maximum objective function is
\begin{equation}
y^*= \delta_R P_{RD} -\frac{(1-P_{SD})P_{SR}}{P_{SD}+ (1-P_{SD})P_{SR}}x -\frac{\delta_R P_{RD}}{(1-\delta_R)\left[ P_{SD}+ (1-P_{SD})P_{SR} \right]}x,
\end{equation}
\begin{equation}
\text{for  } x \leq (1-\delta_R)^2 \left[ P_{SD}+ (1-P_{SD})P_{SR} \right].
\end{equation}

The initial constraint of the $[P_2]$ should also met for $q_R^*=\delta_R$, we
have an additional condition $x < \delta_S (1-\delta_R) \left[ P_{SD}+ (1-P_{SD})P_{SR} \right]$.

Summarizing, the closure obtained by $[P_2]$,
If $\delta_S + \delta_R < 1$ then
\begin{equation}
y^*= \delta_R P_{RD} -\frac{(1-P_{SD})P_{SR}}{P_{SD}+ (1-P_{SD})P_{SR}}x -\frac{\delta_R P_{RD}}{(1-\delta_R)\left[ P_{SD}+ (1-P_{SD})P_{SR} \right]}x
\end{equation}
for
\begin{equation}
x < \delta_S (1-\delta_R) \left[ P_{SD}+ (1-P_{SD})P_{SR} \right].
\end{equation}

If $\delta_S + \delta_R \geq 1$ then
\begin{displaymath}
   y^* = \left\{
     \begin{array}{lr}
      P_{RD} - 2 P_{RD} \sqrt{\frac{x}{P_{SD}+ (1-P_{SD})P_{SR}}} + \frac{x}{P_{SD}+ (1-P_{SD})P_{SR}} \left[P_{RD} - (1-P_{SD})P_{SR} \right]  & \text{ if } x_1 < x \leq x_2 \\
      \delta_R P_{RD} -\frac{(1-P_{SD})P_{SR}}{P_{SD}+ (1-P_{SD})P_{SR}}x -\frac{\delta_R P_{RD}}{(1-\delta_R)\left[ P_{SD}+ (1-P_{SD})P_{SR} \right]}x  & \text{ if }  x \leq x_1
     \end{array}
   \right.
\end{displaymath}

where $x_1=(1-\delta_R)^2 \left[ P_{SD}+ (1-P_{SD})P_{SR} \right]$ and $x_2=\delta_S ^2 \left[ P_{SD}+ (1-P_{SD})P_{SR} \right]$.

Following the same methodology we obtain the solution to the $[P_1]$.
If $\delta_S + \delta_R < 1$ then the closure is
\begin{equation}
\frac{x}{\delta_S \left[ P_{SD}+ (1-P_{SD})P_{SR} \right]} +\frac{(1-P_{SD})P_{SR}x}{(1-\delta_S)P_{RD}\left[ P_{SD}+ (1-P_{SD})P_{SR} \right]} + \frac{y}{(1-\delta_S)P_{RD}} =1,
\end{equation}

for $\frac{(1-P_{SD})P_{SR}}{P_{SD}+ (1-P_{SD})P_{SR}}x+y < P_{RD} (1-\delta_S)\delta_R$.

If $\delta_S + \delta_R \geq 1$ then
if $P_{RD} (1-\delta_S)^2 < \frac{(1-P_{SD})P_{SR}}{P_{SD}+ (1-P_{SD})P_{SR}}x+y \leq P_{RD} \delta_R^2$ the closure is
\begin{equation}
\sqrt{\frac{x}{P_{SD}+ (1-P_{SD})P_{SR}}} + \sqrt{\frac{P_{SR}(1-P_{SD})x}{P_{RD}\left[ P_{SD}+ (1-P_{SD})P_{SR} \right]}+\frac{y}{P_{RD}}} = 1
\end{equation}

if $\frac{(1-P_{SD})P_{SR}}{P_{SD}+ (1-P_{SD})P_{SR}}x+y \leq P_{RD} (1-\delta_S)^2$ the closure is
\begin{equation}
\frac{x}{\delta_S \left[ P_{SD}+ (1-P_{SD})P_{SR} \right]} +\frac{(1-P_{SD})P_{SR}x}{(1-\delta_S)P_{RD}\left[ P_{SD}+ (1-P_{SD})P_{SR} \right]} + \frac{y}{(1-\delta_S)P_{RD}} =1.
\end{equation}

In the previous sections, we obtained the closure of the outer bound of the stability region, which we show below that this closure is achievable. The sufficient conditions for stability for fixed transmission probabilities are given in Proposition \ref{prop:sufficient_conditions}. For $q_S \leq \delta_S$ and $q_R \leq \delta_R$ we have that the saturated throughput for the source and the relay are given by
\begin{equation} \label{eq:mss_qsqr}
\mu_{S}^{s} = q_S (1 - q_R) \left[p_{SD}+(1-p_{SD})p_{SR} \right],
\end{equation}
\begin{equation} \label{eq:mrs_qsqr}
\mu_{R}^{s} = q_R (1 - q_S) p_{RD}.
\end{equation}

From (\ref{eq:mss_qsqr}) we obtain that
\begin{equation} \label{eq:qS_ms}
q_S = \frac{\mu_{S}^{s}}{(1 - q_R ) \left[p_{SD}+(1-p_{SD})p_{SR} \right]}.
\end{equation}
By substituting (\ref{eq:qS_ms}) into (\ref{eq:mrs_qsqr}) we have
\begin{equation} 
\mu_{R}^{s} = q_R \left[ 1 - \frac{\mu_{S}^{s}}{(1 - q_R ) \left[p_{SD}+(1-p_{SD})p_{SR} \right]} \right] p_{RD}.
\end{equation}

After replacing $\mu_{R}^{s}$ with $\lambda_{R,total}=\lambda_{R}+ \frac{(1-p_{SD})p_{SR}} {p_{SD}+(1-p_{SD})p_{SR}}\lambda_{S}$ and $\mu_{S}^{s}$ with $\lambda_S$ then it is identical with the expression in (\ref{eq:P2}) which corresponds to the outer bound. As a result, a point of the saturated throughput can be controlled to any point on the boundary of $\mathcal{R}$ which is given in Proposition \ref{prop:necessary_conditions}. The same argument holds for $\mu_{S}^{s}$ and the expression in (\ref{eq:P1}).

The previous concludes that the closure of the outer bound of the stability region is achievable.

\section{Conclusions} \label{sec:ECOOP_conclusion}
In this paper, we studied the effect of energy constraints on a relay-aided wireless network in which both source and relay have energy harvesting capabilities. The source and the relay nodes also have external arrivals and network-level cooperation is employed, i.e. the relay forwards a fraction of the source's traffic to the destination. 

We derived necessary and sufficient conditions for stability of the above cooperative communication scenario and we also obtained the exact maximum stable throughput region. Interestingly, the closure of the inner and the outer bound is identical. A key insight of this work with an impact on the design of relay-assisted networks with energy limitations is as follows: when the aggregate charging rate is above one and both the source and the relay lie in the intermediate 
traffic regime, then the system has identical performance with the network without energy constraints. Otherwise stated, in the above setting, the energy
limitations are transparent to the network operation, which is also demonstrated by the non-linear behavior of the bound of the maximum stable throughput region.

This work provides a step in connecting information theory and networking by studying the stable throughput region metric. Additionally, it sheds light on the relatively unexplored and important domain of energy harvesting and assesses the effect of that on this important measure.
 
Future work will include the characterization of the stable throughput region using multi-packet reception instead of the erasure channel with collisions. 

\bibliographystyle{IEEEtran}
\bibliography{thesis}

\begin{thebibliography}{10}
\providecommand{\url}[1]{#1}
\csname url@samestyle\endcsname
\providecommand{\newblock}{\relax}
\providecommand{\bibinfo}[2]{#2}
\providecommand{\BIBentrySTDinterwordspacing}{\spaceskip=0pt\relax}
\providecommand{\BIBentryALTinterwordstretchfactor}{4}
\providecommand{\BIBentryALTinterwordspacing}{\spaceskip=\fontdimen2\font plus
\BIBentryALTinterwordstretchfactor\fontdimen3\font minus
  \fontdimen4\font\relax}
\providecommand{\BIBforeignlanguage}[2]{{%
\expandafter\ifx\csname l@#1\endcsname\relax
\typeout{** WARNING: IEEEtran.bst: No hyphenation pattern has been}%
\typeout{** loaded for the language `#1'. Using the pattern for}%
\typeout{** the default language instead.}%
\else
\language=\csname l@#1\endcsname
\fi
#2}}
\providecommand{\BIBdecl}{\relax}
\BIBdecl

\bibitem{b:PappasGlobalsip2013}
N.~Pappas, M.~Kountouris, J.~Jeon, A.~Ephremides, and A.~Traganitis,
  ``Network-level cooperation in energy harvesting wireless networks,'' in
  \emph{IEEE Global Conference on Signal and Information Processing
  (GlobalSIP)}, Dec 2013, pp. 383--386.

\bibitem{paradiso:energy}
J.~A. Paradiso and T.~Starner, ``Energy scavenging for mobile and wireless
  electronics,'' \emph{IEEE Pervasive Computing}, vol.~4, no.~1, pp. 18--27,
  Jan.-Mar. 2005.

\bibitem{abramson:aloha}
N.~Abramson, ``The aloha system -- another alternative for computer
  communications,'' in \emph{Proceedings of AFIPS Conference}, 1970.

\bibitem{jeon:isitstability}
J.~Jeon and A.~Ephremides, ``The stability region of random multiple access
  under stochastic energy harvesting,'' in \emph{Proceedings of IEEE
  International Symposium on Information Theory (ISIT)}, Aug. 2011.

\bibitem{jeon:stability}
------, ``On the stability of random multiple access with stochastic energy
  harvesting,'' \emph{IEEE Journal on Selected Areas in Communications}, to be
  published.

\bibitem{b:Pappas-JCN}
N.~Pappas, J.~Jeon, A.~Ephremides, and A.~Traganitis, ``Optimal utilization of
  a cognitive shared channel with a rechargeable primary source node,''
  \emph{Journal of Communications and Networks (JCN) Special Issue on Energy
  Harvesting in Wireless Networks}, vol.~14, no.~2, 2012.

\bibitem{b:Yates-NOW}
G.~Kramer, I.~Mari\'{c}, and R.~D. Yates, ``Cooperative communications,''
  \emph{Found. Trends Netw.}, vol.~1, no.~3, pp. 271--425, Aug. 2006.

\bibitem{b:Sadek}
A.~Sadek, K.~Liu, and A.~Ephremides, ``Cognitive multiple access via
  cooperation: Protocol design and performance analysis,'' \emph{Information
  Theory, IEEE Transactions on}, vol.~53, no.~10, pp. 3677 --3696, 2007.

\bibitem{b:Rong1}
B.~Rong and A.~Ephremides, ``Protocol-level cooperation in wireless networks:
  Stable throughput and delay analysis,'' in \emph{Modeling and Optimization in
  Mobile, Ad Hoc, and Wireless Networks, 2009. WiOPT 2009. 7th International
  Symposium on}, 2009, pp. 1 --10.

\bibitem{b:Pappas-ISIT}
N.~Pappas, J.~Jeon, A.~Ephremides, and A.~Traganitis, ``Wireless network-level
  partial relay cooperation,'' in \emph{IEEE International Symposium on
  Information Theory (ISIT)}, July 2012.

\bibitem{b:Krikidis_Energy_Harv}
I.~Krikidis, T.~Charalambous, and J.~Thompson, ``Stability analysis and power
  optimization for energy harvesting cooperative networks,'' \emph{IEEE Signal
  Processing Letters}, vol.~19, no.~1, pp. 20--23, 2012.

\bibitem{b:EphremidesHajekUnion}
A.~Ephremides and B.~Hajek, ``Information theory and communication networks: an
  unconsummated union,'' \emph{IEEE Transactions on Information Theory},
  vol.~44, no.~6, pp. 2416--2434, Oct 1998.

\bibitem{tsybakov:ergodicity}
B.~S. Tsybakov and V.~A. Mikhailov, ``Ergodicity of a slotted aloha system,''
  \emph{Problems of Information Transmission}, vol.~15, no.~4, pp. 301--312,
  1979.

\bibitem{rao:stability}
R.~Rao and A.~Ephremides, ``On the stability of interacting queues in a
  multi-access system,'' \emph{IEEE Trans. on Inform. Theory}, vol.~34, no.~5,
  pp. 918--930, Sep. 1988.

\bibitem{Szpankowski:stability}
W.~Szpankowski, ``Stability conditions for some distributed systems: Buffered
  random access systems,'' \emph{Adv. in App. Prob.}, vol.~26, no.~2, pp.
  498--515, Jun. 1994.

\bibitem{b:Rong3}
B.~Rong and A.~Ephremides, ``On stability and throughput for multiple access
  with cooperation,'' in \emph{under Submission}.

\bibitem{b:Loynes}
R.~Loynes, ``The stability of a queue with non-independent inter-arrival and
  service times,'' \emph{Proc. Camb. Philos.Soc}, vol.~58, no.~3, pp. 497--520,
  1962.

\bibitem{b:VerduAnantharam}
V.~Anantharam and V.~S., ``Bits through queues,'' in \emph{IEEE International
  Symposium on Information Theory}, Jun 1994.

\bibitem{b:Gallager76}
R.~Gallager, ``Basic limits on protocol information in data communication
  networks,'' \emph{IEEE Transactions on Information Theory}, vol.~22, no.~4,
  pp. 385--398, Jul 1976.

\end{thebibliography}

\end{document}